\documentclass[conference,letterpaper]{IEEEtran}

\usepackage{graphicx} 

\usepackage{graphicx,enumitem}
\usepackage[colorlinks=true, allcolors=blue]{hyperref}
\usepackage{caption}
\usepackage{subcaption}
\usepackage{amsthm}
\usepackage{amsfonts,amsthm,amssymb}
\usepackage{algorithm}
\usepackage{array}
\usepackage{textcomp}
\usepackage{balance}
\usepackage{stfloats}
\usepackage{color}
\usepackage{url}
\usepackage{verbatim}
\usepackage{graphicx}
\usepackage{cite}

\usepackage{xcolor}

\usepackage[noend]{algpseudocode}

\usepackage[utf8]{inputenc} 
\usepackage[T1]{fontenc}
\usepackage{url}
\usepackage{ifthen}
\usepackage{cite}
\usepackage[cmex10]{amsmath}

\usepackage{mleftright}       
\mleftright   

\newtheorem{theorem}{Theorem}

\newtheorem{lemma}{Lemma}
\newtheorem{remark}{Remark}
\newtheorem{corollary}{Corollary}
\newtheorem{proposition}{Proposition}

\usepackage{xspace}
\usepackage{bbm}
\usepackage{mathrsfs}

\newcommand{\Hc}{\mathcal{H}}

\newcommand{\Rt}{{\tilde{R}}}

\newcommand{\Vt}{{\tilde{V}}}

\newcommand{\pt}{{\tilde{p}}}

\def\l{\lambda}

\newcommand{\U}{\mathrm{Unif}}

\def\textiid{i.i.d.\@\xspace}
\newcommand\iid{\ifmmode\text{ i.i.d. } \else \textiid \fi}

\newcommand{\ind}{\boldsymbol{1}}

\allowdisplaybreaks
\begin{document}
\title{Distributed Multiple Testing with False Discovery Rate Control in the Presence of Byzantines} 


\author{%
  \IEEEauthorblockN{Daofu Zhang}
  \IEEEauthorblockA{Electrical and Computer Engineering \\
                    University of Utah\\
                    Email: daofu.zhang@utah.edu}
                      \bigskip\IEEEauthorblockN{Yu Xiang}
  \IEEEauthorblockA{Electrical and Computer Engineering\\ 
  University of Utah\\
                    Email: yu.xiang@utah.edu}
  \and
  \IEEEauthorblockN{Mehrdad Pournaderi}
  \IEEEauthorblockA{Microbiology and Immunology\\ 
                    University at Buffalo--SUNY\\
                    Email: mehrdadp@buffalo.edu}
                    \bigskip
  \IEEEauthorblockN{Pramod Varshney}
  \IEEEauthorblockA{Electrical Engineering and Computer Science\\ 
  Syracuse University\\
                    Email: varshney@syr.edu}
                                    
}

\maketitle


\begin{abstract}
   This work studies distributed multiple testing with false discovery rate (FDR) control in the presence of Byzantine attacks, where an adversary captures a fraction of the nodes and corrupts their reported $p$-values. We focus on two baseline attack models: an oracle model with the full knowledge of which hypotheses are true nulls, and a practical attack model that leverages the Benjamini-Hochberg (BH) procedure locally to classify which $p$-values follow the true null hypotheses. We provide a thorough characterization of how both attack models affect the global FDR, which in turn motivates counter-attack strategies and stronger attack models. Our extensive simulation studies confirm the theoretical results, highlight key design trade-offs under attacks and countermeasures, and provide insights into more sophisticated attacks.

\end{abstract}

\section{Introduction}

We consider the problem of testing multiple hypotheses over a network with a central agent, in the presence of Byzantine attacks; the hypotheses may come from testing multiple local test data samples (e.g., outlier detection). \emph{An adversary (or adversarial agent) can capture a fraction of the nodes and launch a Byzantine attack}. As a consequence, the attacked nodes will report statistics altered by an adversary to the central decision-making unit, thereby corrupting the statistical properties of the data. Specifically, we focus on the global performance under the false discovery rate (FDR) control~\cite{benjamini1995,benjamini2001control,efron2001empirical,genovese2002operating,storey2002direct}, a widely-used statistical measure that quantifies the expected proportion of false rejections. Our work is partially motivated by the recent line of works on outlier detection from the multiple testing perspective (e.g.,~\cite{bates2023testing,kaur2022idecode}), where the goal is to perform out-of-distribution detection under FDR control. Our setup can therefore help connect multiple testing frameworks and distributed settings under adversarial attacks, including distributed intrusion detection systems~\cite{tlili2024exhaustive}, identifying fraud patterns through collaborative analysis~\cite{hu2023collaborative}, and environmental monitoring using sensor networks~\cite{zhang2008framework}.


Without adversarial attacks, the distributed multiple testing problem under FDR control has been studied from various perspectives in the literature~\cite{ray2007novel,ray2011false,ermis2006detection,ermis2009distributed,Ramdas2017b,pournaderi2023sample,pournaderi2023large,golz2022multiple}. In the pioneering works~\cite{ray2007novel,ray2011false,ermis2006detection,ermis2009distributed}, the authors have investigated the distributed sensor networks under a broadcast model, where each sensor is
allowed to broadcast its decision to the entire network. More recently, it has been shown that FDR control can be achieved in multi-hop network settings~\cite{Ramdas2017b}. Along the same lines, the communication-efficiency perspective has been studied in the finite-sample~\cite{pournaderi2023sample} and asymptotic~\cite{pournaderi2023large} regimes. A similar theme has been investigated in~\cite{vempaty2014false}, yet under a completely different formulation from this work.


The objective of this study is to understand the impact of Byzantine attacks in terms of controlling the global FDR over the entire network. Our contributions are threefold. First, we introduce two baseline attack models. One is the oracle setting where the attacker has knowledge of the underlying hypothesis (true null vs. false null hypothesis) of each $p$-value under attack. This baseline model is the ideal setting that can not be realized in real-world scenarios. This motivates us to study a practical attack model that relies on using the celebrated Benjamini-Hochberg (BH) procedure~\cite{benjamini1995}, which controls FDR, as a classification technique. Then, we formally characterize the cost in terms of FDR under both models and develop counter-attack schemes along with stronger attack models (\emph{enhanced BH‐classifier attack} and \emph{shuffling attack}), building on our baseline model.  Lastly, we carry out extensive experimental studies to verify our theoretical findings as well as explore other potential attack strategies. 


\section{Problem Formulation}
\label{sec:form}
 Suppose that there are $n$ null hypotheses distributed over a network with $d$ nodes along with one central agent, where each node needs to test $n/d$ hypotheses and we assume $n/d$ is an integer for simplicity of presentation throughout this work. 

Let $\mathsf{H_{0,i}, 1 \leq i\leq n}$, denote the null hypotheses and each node performs their test based on the test statistics $X_i, 1 \leq i\leq n$. Let $p_i = 2 \cdot \min\left\{ \mathsf{F}_{\mathsf{H_{0,i}}}(X_i),\, 1 - \mathsf{F}_{\mathsf{H_{0,i}}}(X_i) \right\}, 1 \leq i \leq n,$ denote the $p$-values computed for the test statistics, where $\mathsf{F}_{\mathsf{H_{0,i}}}$ is the CDF of $X_i$ under $\mathsf{H_{0,i}}$. Let $\Hc_0$ denote the set consisting of all the true null hypotheses (or true nulls for short) and we assume that the cardinality of $\Hc_0$ is $n_0$, that is, $|\Hc_0|=n_0$. Throughout this work, we assume that all the $n_0$ $p$-values under the null hypothesis are independent and they are independent of the non-null $p$-values, which is the classical assumption in the FDR literature; even though some of our results can be readily extended to some dependent settings, we leave the comprehensive treatment for future work.

The FDR measures the expected incorrect rejections of true null hypotheses, among all rejected hypotheses: 
\begin{equation}
\text{FDR} = \mathbb{E}\left[\frac{V}{R\vee 1}\right],\nonumber
\end{equation}
where $V$ is the number of false rejections, $R$ is the total number of rejections, and $a\vee b:=\max\{a,b\}$. The power of a multiple testing procedure is the expected true positive proportion, defined as $\text{power} = \mathbb{E}\left[\frac{R-V}{n_1\vee 1}\right]$, where $n_1=n-n_0$. 

The attacker captures a fraction $\lambda$ of nodes that have $m$ $p$-values in total, $\{p_i\}_{i\in \Hc^a}$ with $|\Hc^a|=m$, and carries out the attack by changing them to $\{\pt_i\}_{i\in \Hc^a}$ in an adversarial way. Note that this implies that the fraction is $\lambda=m/n$. Among all the  $m=m_0+m_1$ $p$-values, there are $m_0$ true nulls (indexed by $\Hc^a_0$ with $|\Hc^a_0|=m_0$) and $m_1$ non-nulls (indexed by $\Hc^a_1$ with $|\Hc^a_1|=m_1$).  Throughout this work, we assume that the nodes being attacked need to send $\{\pt_i\}_{i\in \Hc^a}$ to the central server (i.e., the central agent will receive $n$ $p$-values from all of the $d$ nodes). After receiving all the $p$-values sent by the nodes in the network, the central agent runs the BH procedure globally to make $\tilde{R}$ rejections, we assume the target FDR level $q>0$ throughout the work and $\tilde{V}$ denotes the total number of false rejections, leading to the FDR after the attack:
\begin{equation}
\text{FDR}_{\text{attack}} = \mathbb{E}\left[\frac{\Vt}{\Rt\vee 1}\right].\nonumber
\end{equation} 
In the following sections, we will start by analyzing an oracle attack setting and then study one practical attack model (based on the BH procedure) as well as counter-attack schemes. 

\section{Oracle Attack Analysis}
\label{sec:oracle}
 For simplicity, we start with analyzing the setting when the attacker captures one node that has $m$ $p$-values; this implies that $m=n/d$ according to Section~\ref{sec:form}. The BH procedure conducted at the central agent consists of three steps:
\begin{itemize}\setlength\itemsep{0.06em}
    \item[(1)] Arrange all $p$-values in increasing order $p_{(1)}\leq p_{(2)} \leq \dots\leq p_{(n)}$.
    \item[(2)] Compute thresholds for each $p$-value as $\tau_i := q\cdot(i/n)$, where $q$ is the desired FDR level.
    \item[(3)] Reject all hypotheses with $p$-values $p_{(i)}\leq p_{(i_0)}$, where $i_0:=\max\{i: p_{(i)}\le \tau_i\}$.
\end{itemize} 

We first introduce the \emph{oracle attack model}, where we assume that the attacker knows which $m_0$ $p$-values are true nulls.

\smallskip
\noindent{\bf Oracle attack model:} \emph{The attacker changes all of the $m_0$ true null $p$-values to $0$ and changes the rest of the $m_1$ non-null $p$-values to $1$. }
\medskip


We denote the FDR under the oracle attack model by $\text{FDR}^*_{\text{attack}}$. 
Now we are ready to state our first main result. 
\begin{theorem}
\label{thm1}
    Suppose the attacker captures one node with $m$ $p$-values, and carries out the oracle attack. Then
\begin{align}    \text{FDR}^*_{\text{attack}}=m_0\cdot\mathbb{E}\left[\frac{1}{\Rt\vee 1}\right]+\frac{q(n_0-m_0)}{n},
\end{align}  
    when the BH procedure is applied at the central agent. 
\end{theorem}
\begin{proof}
    Recall that we use $\Hc^a_0$ to denote the set of true nulls the attacker can access and $|\Hc^a_0|=m_0$ denotes the number of $p$-values changed to $0$ by the attacker. Noting that the $0$ $p$-values get rejected by the BH procedure since $q>0$ by our assumption, we get
\begin{align*}
\text{FDR}^*_{\text{attack}} & = \mathbb{E} \left[ \sum_{i \in \Hc^a_0} \frac{\tilde{V}_i}{\tilde{R} \lor 1} \right] 
+ \mathbb{E} \left[ \sum_{i \in \Hc_0 \setminus \Hc^a_0} \frac{\tilde{V}_i}{\tilde{R} \lor 1} \right] \\ 
& = \mathbb{E}\left[\frac{m_0}{\Rt \vee 1}\right] + \frac{q}{n}\sum_{i\in \Hc_0 \setminus \Hc^a_0}\mathbb{E}\left[\frac{\ind\{p_i\leq q \Rt/n\}}{(q/n) (\Rt \vee 1)}\right] \\
& = \mathbb{E}\left[\frac{m_0}{\Rt \vee 1}\right]  + q\frac{(n_0 - m_0)}{n},
\end{align*}
where $\tilde{V}_i = \ind\{p_i\leq q \Rt/n\}$, $\Hc_0 \setminus \Hc^a_0$ denotes the indices of true nulls that the attacker has not touched, and the last equality holds according to~\cite[Lemma~3.2]{blanchard2008} (also see~\cite{ramdas2019unified}).
\end{proof}
It is well-known that the BH procedure guarantees that $\text{FDR}=q(n_0/n)$. From Theorem~\ref{thm1}, it is straightforward to see that  $\text{FDR}^*_{\text{attack}}\ge q(n_0/n)$, which follows from the fact that $\tilde{R}\le n$. Therefore, the oracle attack will always result in an increase in the FDR. 
\begin{remark}
    In the proof of Theorem~\ref{thm1}, the second term holds regardless of the attack strategy. For the first term, one  always has $\mathbb{E} \left[ \sum_{i \in \Hc^a_0} \frac{\tilde{V}_i}{\tilde{R} \lor 1} \right] 
        \leq \mathbb{E}\left[\frac{m_0}{\Rt \vee 1}\right]$ as $\tilde{V}_i\leq 1$. Thus, Theorem~\ref{thm1} holds with inequality for any attack strategy and is achievable by the oracle attack. However, $\tilde{R}$ in the bound still depends on the attack model. To make the bound independent of the attacker strategy, one can upper bound the first term with 
\[
    \mathbb{E} \left[ \frac{m_0}{\tilde{R}(P_{a} \to 1) \lor 1} \right],
\]
where $\tilde{R}(P_{a} \to 1)$ denotes the number of rejections when all the attacked $p$-values are set to $1$. This bound cannot be achieved (except for $\Hc^a_0=\varnothing$) since the attacker's true nulls are considered rejected in the numerator and not rejected in the denominator.
\end{remark}

The same analysis applies when the attacker captures a fraction of nodes that have $m$ $p$-values in total, in which $m_0$ of them are true nulls.
\begin{corollary}

    Suppose that the attacker captures $\lambda d$ nodes with a total of $m$ $p$-values, where each node has $n/d$ local $p$-values. Each node carries out the oracle attack. Then
  \begin{align}
    \text{FDR}^*_{\text{attack},\l}=\text{FDR}^*_{\text{attack}}
    \end{align}  
    when the BH procedure is applied at the central agent.
\end{corollary}

\section{BH-classifier Attack Model}
\label{sec:class}
\noindent Again, suppose that the attacker captures one node that has $m$ $p$-values, but the attacker does not know which ones are the true nulls or true non-nulls; we will then extend this to the multiple-node case. We now introduce a natural and practical attack model.

\noindent{\bf BH-classifier attack model:} 
\begin{itemize}\setlength\itemsep{0.06em}
    \item[(1)] \emph{The attacked node first applies the BH procedure on the $m$ $p$-values at hand by sorting them in ascending order ($p_{(1)}\leq p_{(2)}\leq\dots\leq p_{(m)}$).}
    \item[(2)] \emph{It returns an index $i_0:=\max\{i: p_{(i)}\le q(i/m)\}$ and classifies all $p_{(i)}$ where $i\leq i_0$ as non-nulls, and classifies the remaining $p$-values as true nulls.}
    \item[(3)]  \emph{It changes all those $p$-values classified as true nulls to be $0$, and classified as non-nulls to be $1$. }
\end{itemize}


Suppose that the attacker captures $\lambda d$ nodes with a total of $m$ $p$-values, where each node has $n/d$ local $p$-values, and let $m_{0,i}$ be the number of true nulls at node~$i$. Each node carries out the BH-classifier attack.

\begin{theorem}
    
    In this distributed setting, we have 
\begin{align*}
\text{FDR}_{\text{attack},\lambda}
&\le \sum_{i=1}^{\lambda d} m_{0,i}\,\mathbb{E}\!\left[\frac{1 - (qd/n)\,R_i}{\Rt \vee 1}\right] \\
&\quad\quad\quad\quad\quad\quad+ \frac{\bigl(n_0 - \sum_{i=1}^{\lambda d}m_{0,i}\bigr)\cdot q}{n}\,
\end{align*}
by applying the BH procedure at the central agent, where $R_i$ denotes the number of local rejections at the $ith$ node. 
\end{theorem}

The analysis of this distributed setting is a straightforward extension of the centralized setting (Theorem~\ref{thm:2}) and is thus omitted due to space limitations. In the following, we state and prove the result in the centralized setting. 


Suppose the attacker makes $R_a$ rejections after applying the BH algorithm in the classification step. We can upper bound the corresponding $\text{FDR}_{\text{attack}}$ as follows.

\begin{theorem}
\label{thm:2}

Suppose the attacker captures one node with $m$ $p$-values, and carries out the BH-classifier attack. Then
    \begin{align} \text{FDR}_{\text{attack}}&=\sum_{i\in \Hc^a_0}\mathbb{E}\left[\frac{1-(q/m) R_{a}(p_i\to 0)}{\Rt(p_i\to 1)}\right]+\frac{q(n_0-m_0)}{n} \nonumber\\
    &\le m_0\, \mathbb{E} \left[ \frac{1-(q/m)R_{a}}{\Rt \vee 1}\right]+\frac{q(n_0-m_0)}{n}\label{eq:BH-attack}
    \end{align} when the BH procedure is applied at the central agent, where $\tilde{R}(p_i \to 1)$ and $R_a(p_i \to 1)$ denote the new rejection counts after replacing $p_i$ with $1$. 
\end{theorem}
We can easily see that the upper bound in~\eqref{eq:BH-attack} can be further upper bounded by the oracle FDR given in Theorem~\ref{thm1}. Furthermore, when $R_a/m \approx 0$, the upper bound in~\eqref{eq:BH-attack} is close to the oracle FDR (see Experiment~1 in Section~\ref{sec:exp} for numerical examples). Thus, the BH-classifier attack model can be viewed as a practical baseline, which serves as a surrogate for the oracle attack model.


\begin{proof}

Recall that we use $\Hc^a_0$ to denote the set of all the true nulls that the attacker has at hand and $|\Hc^a_0|=m_0$. Then, we can express FDR after the attack as follows,
\begin{equation}
\text{FDR}_{\text{attack}} = \mathbb{E} \left[ \sum_{i \in \Hc^a_0} \frac{\tilde{V}_i}{\tilde{R} \lor 1} \right] 
+ \mathbb{E} \left[ \sum_{i \in \Hc_0 \setminus \Hc^a_0} \frac{\tilde{V}_i}{\tilde{R} \lor 1} \right],\label{eq:attack}
\end{equation}
where $\tilde{V}_i = \ind\{p_i\leq q \Rt/n\}$.
For each $i \in \Hc^a_0$ in the first term, 
\begin{align}
  \mathbb{E} \left[  \frac{\tilde{V}_i}{\tilde{R} \lor 1} \right] 
    &= \mathbb{E}\left[\frac{\ind\{p_i > q R_{a}/m\}}{\Rt \vee 1}\right]\label{pqm} \\
    &= \mathbb{E}\left[\frac{\ind\{p_i > q R_{a}(p_i\to 0)/m\}}{\Rt(p_i\to 1) }\right]\label{lem1},
\end{align}
where \eqref{pqm} comes from the fact that if and only if $p_i>qR_a/m$, $p_i$ will not be rejected by the attacker's BH classification and $\tilde{p_i}$ will be $0$ accordingly which will make $\tilde{V_i}$ to be $1$. And \eqref{lem1} holds because of Lemma~\ref{lem:1}, and the fact that $\Rt(p_i\to 1)=\Rt \vee 1$ when $\ind\{p_i > q R_{a}/m\}=1$. 
Hence, the first term in~\eqref{eq:attack} can be expressed as 
\begin{align*}
&\hspace{-1em}\mathbb{E} \left[ \sum_{i \in \Hc^a_0} \frac{\tilde{V}_i}{\tilde{R} \lor 1} \right] 
    = \sum_{i\in \Hc^a_0}\mathbb{E}\left[\frac{\ind\{p_i > q R_{a}(p_i\to 0)/m\}}{\Rt(p_i\to 1) }\right].
\end{align*}
Conditioning on all $p$-values except $p_i$ which is represented as $\mathcal{F}_i = \sigma(\{p_1, \dots, p_{i-1}, p_{i+1}, \dots, p_n\})$, we get
\begin{align}
&\sum_{i\in \Hc^a_0}\mathbb{E}\left[\frac{\ind\{p_i > q R_{a}(p_i\to 0)/m\}}{\Rt(p_i\to 1) }\,\bigg|\, \mathcal{F}_i\right]\label{eq:conditional1}
\\
&= \sum_{i\in \Hc^a_0}\frac{\mathbb{E}\left[\ind\{p_i > q R_{a}(p_i\to 0)/m\}\,\bigg|\, \mathcal{F}_i\right]}{\Rt(p_i\to 1) }\\
&=\sum_{i\in \Hc^a_0}\frac{1-(q/m) R_{a}(p_i\to 0)}{\Rt(p_i\to 1)},
\end{align}
where we can move $\Rt(p_i\to 1)$ outside the expectation in~\eqref{eq:conditional1} because it is $\mathcal{F}_i$-measurable,
and the last equality comes from the fact that the $p$-value under true null follows $\U[0,1]$. 
We now use the tower property to bound the first term in~\eqref{eq:attack},
\begin{align*}
    &\sum_{i\in \Hc^a_0}\mathbb{E}\left[\frac{\ind\{p_i > q R_{a}(p_i\to 0)/m\}}{\Rt(p_i\to 1) }\right]\label{loop}\\
    &=\sum_{i\in \Hc^a_0}\mathbb{E}\left[\mathbb{E}\left[\frac{\ind\{p_i > q R_{a}(p_i\to 0)/m\}}{\Rt(p_i\to 1) }\,\bigg|\, \mathcal{F}_i\right]\right]\\
     &=\sum_{i\in \Hc^a_0}\mathbb{E}\left[\frac{1-(q/m) R_{a}(p_i\to 0)}{\Rt(p_i\to 1)}\right]
                  \leq  m_0\, \mathbb{E} \left[ \frac{1-(q/m)R_{a}}{\Rt \vee 1}\right],
\end{align*}
where the last line follows from the fact that the null $p$-values are \iid $\U[0,1]$ (hence exchangeable), $R_{a}(p_i\to 0)\ge R_{a}$, and $\Rt(p_i\to 1)\le \Rt$. For the second term in~\eqref{eq:attack}, we have
\begin{align*}
    \mathbb{E} \left[ \sum_{i \in \Hc_0 \setminus \Hc^a_0} \frac{\tilde{V}_i}{\tilde{R} \lor 1} \right]
    &=  \sum_{i\in \Hc_0 \setminus \Hc^a_0}\mathbb{E}\left[\frac{ \ind\{p_i\leq q \Rt/n\}}{\Rt \vee 1}\right] \\ 
    & = q\frac{(n_0 - m_0)}{n},
\end{align*}
where we noted that $\pt_i=p_i$ for those $i \in \Hc_0 \setminus \Hc^a_0$ and the last equality follows from the same argument as in Theorem~\ref{thm1}. 

Putting everything together, we get
\begin{align*}
\text{FDR}_{\text{attack}} 
&\leq m_0\, \mathbb{E} \left[ \frac{1-(q/m)R_{a}}{\Rt \vee 1}\right]+\frac{(n_0-m_0)\cdot q}{n}.
\end{align*}
\end{proof}


\begin{lemma}\label{lem:1}
    For each $i \in \Hc^a_0$, we have 
    \begin{equation*}
        \mathbf{1}\{p_i >qR_a(p_i \to 0)/m \}=\mathbf{1}\{p_i >qR_a/m \}.
    \end{equation*}
\end{lemma}
\begin{proof}
First consider the case when $p_i\leq qR_a/m$, then this $p$-value is already rejected. Pushing it to $0$ will not change the total rejection $R_a$, which means $p_i\leq qR_a(p_i \to 0)/m$. 
    
    Now consider the other case when $p_i>qR_a/m$, without loss of generality, we assume $p_i$ is the $i$th smallest $p$-value. Since $p_i$ was not rejected, we have $p_i>qi/m$. Also, note that $R_a(p_i \to 0) \leq i$ because sending $p_i$ to $0$ will change the threshold only for $p$-values smaller than $p_i$. Hence $p_i>qR_a(p_i \to 0)/m$. The claimed equality holds for both cases.
    
\end{proof}

\subsection{Counter-attack strategy}
\label{sec:counter}
Suppose that the central server knows (1) the attacker is implementing the BH-classifier attack model, and (2) which nodes are part of the Byzantine (i.e., nodes that have been captured by the attacker). It is natural to ask if it is possible to mitigate the FDR loss. It turns out that the FDR can be controlled by implementing a simple scheme as follows.

\smallskip
\noindent{\bf Counter-attack scheme:} \emph{For each of the $p$-values that have been set to $0$ by the attacker, the central server replaces it with a sample drawn from $\U[0, 1]$. }
\smallskip

\begin{proposition}
    This counter-attack strategy controls FDR. 
\end{proposition}
Here we leverage the fact that true null $p$-values are distributed according to $\U[0, 1]$ and altering the non-null $p$-values won't affect the FDR. This is by no means the only possible counter-attack scheme and we leave the other effective ones for future work.
\subsection{Two stronger attack models}
Since the BH-classifier attack model fails to affect FDR when the central server knows the attack scheme and applies the simple counter-attack scheme. In this subsection, we introduce two attack models which are hard to be counter-attacked by the central server. 
\smallskip
\begin{itemize}[nosep]
    \item \textbf{Enhanced BH-classifier attack model}: \emph{The attacked node first applies BH to classify the local $p$-values. Then, the ones that are classified as nulls are scaled to the range of the classified non-nulls and vice versa.}
    \item \textbf{Shuffling attack model}: \emph{The attacker randomly permutes the indices of all its local $p$-values and then sends $p$-values with the new indices to the central server.}
\end{itemize}
\smallskip

 In the enhanced BH-classifier model, the idea is to hide the identities of the classified nulls into the classified non-nulls. The shuffling attack model decouples each attacked $p$-value and its corresponding hypothesis, and the global BH threshold does not change. Specifically, each $p$-value under attack is true null with probability $m_0/m$ and non-null with probability $m_1/m$; one can upper bound the $\text{FDR}_{\text{attack}}$ as $m_0\cdot(\frac{m_0q}{mn}+\frac{m_1}{m}\mathbb{E}\left[\frac{1}{R} \right])$.

\section{Experimental Results}
\label{sec:exp}

In this section, we compare the $\text{FDR}^*_{\text{attack}}$ and $\text{FDR}_{\text{attack}}$ by conducting a series of experiments. The total number of hypotheses is fixed at $n = 10^4$ ($n_0$ true nulls and $n_1 = n - n_0$ non-nulls) and the level $q$ is fixed at $0.05$. The attacker has $m$ $p$-values in hand ($m_0$ true nulls and $m_1=m-m_0$ non-nulls). Adversarial modifications are applied as specified in the two attack models: oracle attack and BH-classifier attack. For all the experiments, the $p$-values are generated as follows:
\begin{itemize}[nosep]
    \item True null hypothesis: The test statistics are sampled from $N(0, 1)$. The two-sided $p$-values are calculated as $p_i = 2 (1 - \Phi(|X_i|))$, 
    where $\Phi$ is the cumulative distribution function of the standard normal distribution.
    \item Alternative hypothesis: The test statistics are sampled from $N(\mu, 1)$, where $\mu \sim \text{Unif}(1.0, 1.5)$.
\end{itemize}

\smallskip
\noindent\underline{\bf Exp. 1: FDR under oracle vs. BH-classifier attacks}
\smallskip

\noindent{\bf Setting 1: Varying $n_0$ and $n_1$.} For fixed attacker fraction $m/n$=$0.2$, we analyze the impact of changing the proportion of true nulls ($n_0$) and non-nulls ($n_1$) while keeping $n = 10^4$.

\smallskip
\noindent{\bf Setting 2: Varying $m$.} For fixed proportion of true nulls and non-nulls ($n_0=8000,n_1=2000$), we evaluate the effect of varying the number of $p$-values modified by the attacker. We compute $\text{FDR}^*_{\text{attack}}$ and $\text{FDR}_{\text{attack}}$ to compare the gap as $m_0$ and $m_1$ increase.




\begin{figure}[htbp]
    \centering
    \begin{subfigure}[t]{1.0\linewidth} 
        \centering
        \includegraphics[width=\linewidth]{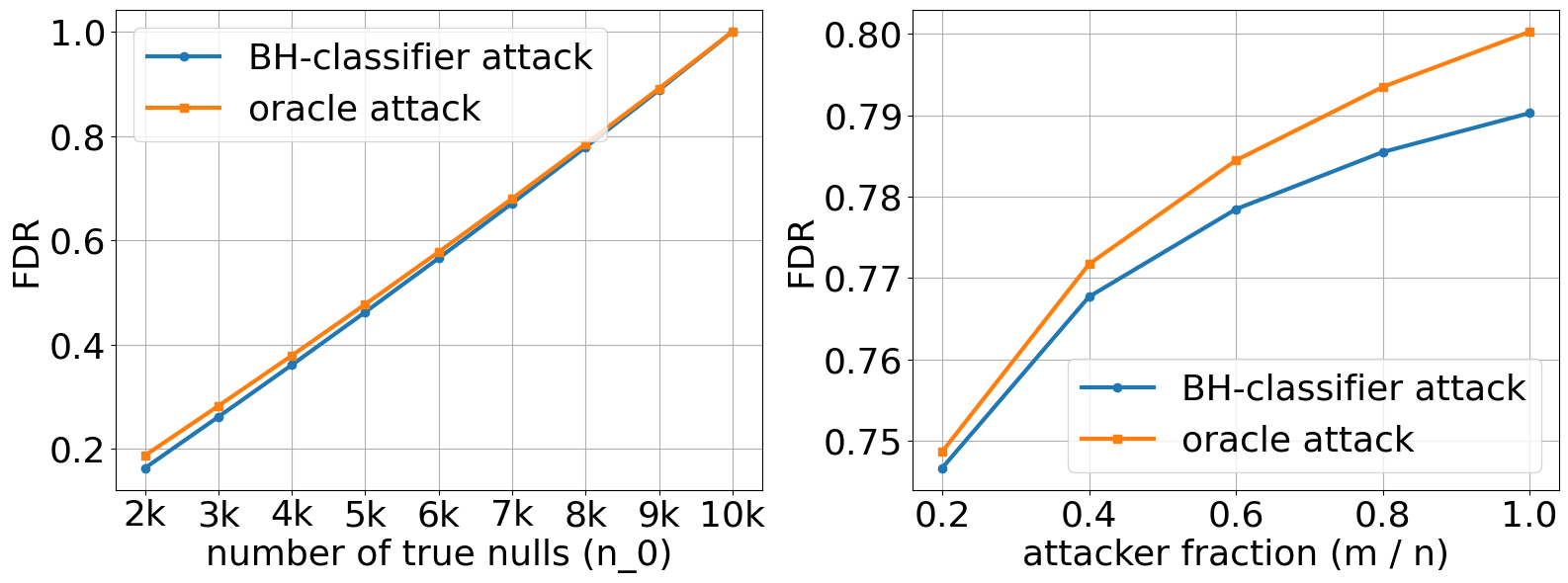} 
        \label{fig:fig1}
    \end{subfigure}
    \hspace{0.05\linewidth} 
\vspace{-1.2em}
    \caption{Exp. 1. In both settings, the gap between $\text{FDR}^*_{\text{attack}}$ and $\text{FDR}_{\text{attack}}$ remains negligible overall.}
    \label{fig:side_by_side_fixed}
    \vspace{-1em}
\end{figure}
This experiment shows that the BH-classifier attack incurs almost the same amount of degradation of FDR as the oracle model in these settings, implying that the BH-classifier attack model, without the information of which $p$-values are true nulls, can be viewed as a practical baseline that approximates the oracle setting very well.

\smallskip
\noindent\underline{\bf Exp. 2: Counter-attack strategy}
\smallskip

\noindent In this experiment, we evaluate the effectiveness of a counter-attack strategy employed by the central server to mitigate the impact of adversarial attacks. Focusing on attacking one node with $m$ $p$-values, we assume that the central server knows (1) the attacker is implementing the BH-classifier attack model, and (2) which node is under attack. 

We empirically compare $\text{FDR}_{\text{attack}}$ with and without applying this counter-attack scheme as mentioned in Section~\ref{sec:counter} along with removing all the $0$ $p$-values. The data generation process is the same as in previous experiments. The empirical $\text{FDR}_{\text{attack}}$ is estimated by averaging over $10^4$ trials for different numbers of $n_0$ while keeping $m=2000$.
\begin{itemize}[nosep]
    \item \textbf{Without counter-attack}: The central server directly applies the BH procedure to all the received $p$-values including the adversarially modified $p$-values.
    \item \textbf{With counter-attack}: (I) The central server replaces all $0$ $p$-values received from the attacker with independently samples from $\U(0, 1)$ and then applies BH over all the $p$-values. (II) The central server simply removes all those $0$ $p$-values and then applies BH over the remaining ones.
\end{itemize}

\begin{figure}[h!]
    \centering
    
    \begin{subfigure}{0.5\textwidth}
        \centering
        \includegraphics[width=\linewidth]{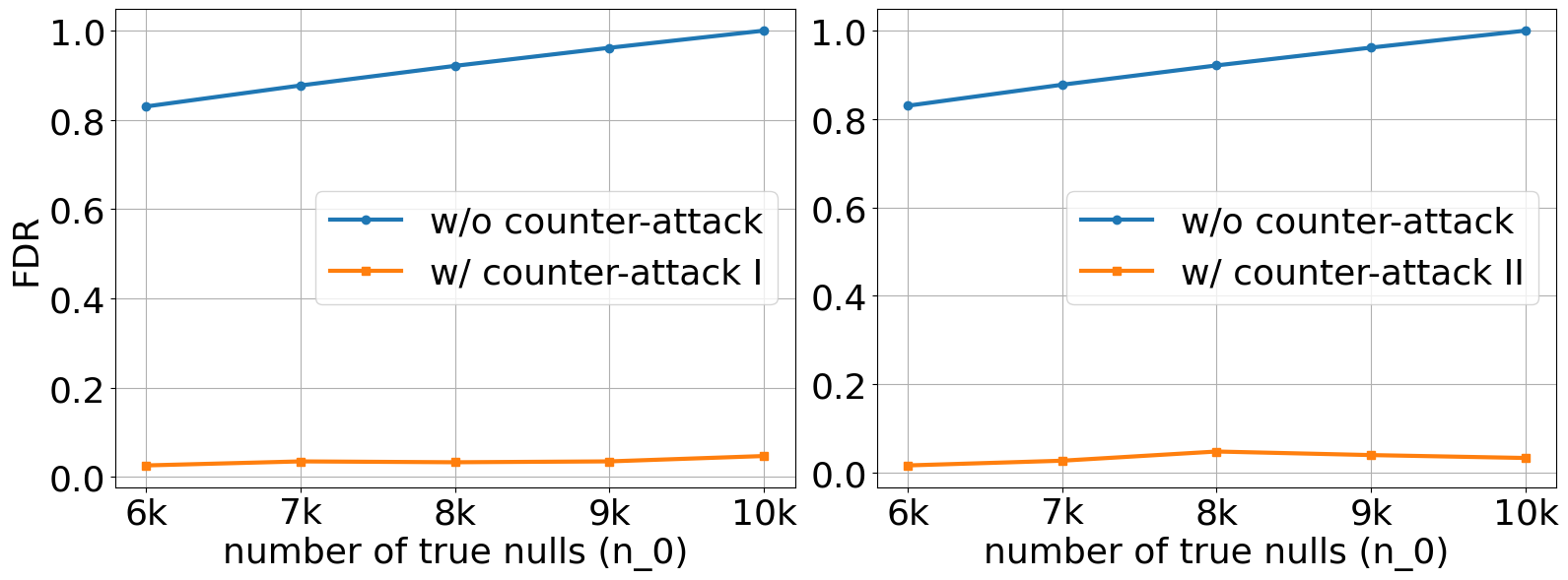}
      
        \label{fig:example4}
    \end{subfigure}
    \hfill
    \vspace{-1.2em}
    \caption{Exp. 2. Effectiveness of counter-attack methods. Comparison of FDR with vs. without applying two types of counter-attack schemes. Counter-attack I (left) and II (right).}
    \label{fig:three_in_one_row}
    \vspace{-1em}
\end{figure}


\smallskip
\noindent\underline{\bf Exp. 3: Two stronger attack models}

\smallskip

\noindent We illustrate the two stronger attack models as mentioned in the previous section; it is important to note that the two counter-attack methods in Exp. 2 do not work for these two attacks, since the nulls and non-nulls are indistinguishable from the central server's perspective. The first plot shows how the enhanced BH-classifier attack significantly increases the FDR as the number of true null hypotheses ($n_0$) increases but subsequently decreases when
$n_0$ becomes excessively large. To explain this phenomenon, we found in our experiments that when $n_0$ becomes too large, the attacker's local BH-classification step will make very few rejections. Consequently, the rescaling step will alter the majority of local non-null $p$-values to smaller values, which ultimately helps the central agent make more correct rejections. The second plot, focusing on the shuffling attack, reveals a less pronounced increase in FDR. Although FDR still grows with $n_0$ and $m/n$, the shuffling attack's impact is weaker and less dynamic due to its lack of strategic manipulation. Together, the plots demonstrate that the enhanced BH-classifier attack is more effective in exploiting the hypothesis testing process to compromise FDR, especially at larger attacker fractions. 


\begin{figure}[htbp]
    \centering
    \begin{subfigure}[t]{1.0\linewidth} 
        \centering
        \includegraphics[width=\linewidth]{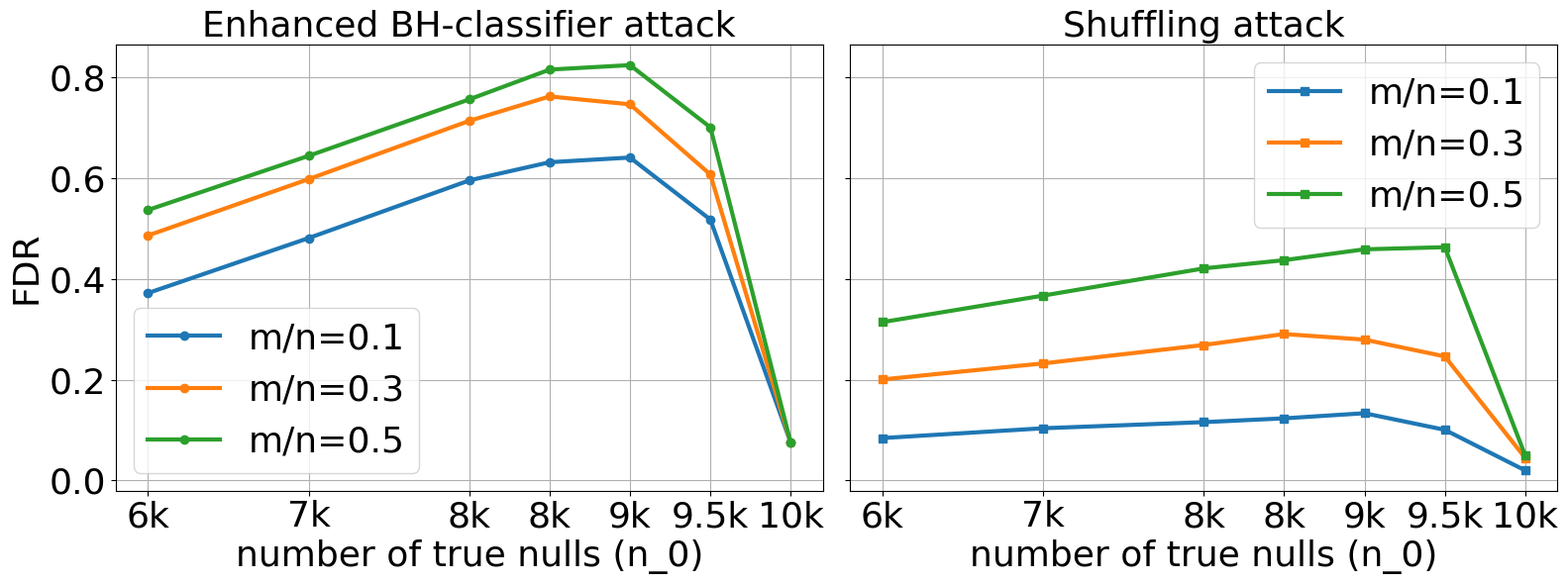} 
        \label{fig:fig1}
    \end{subfigure}
    \hspace{0.05\linewidth} 
    \vspace{-1.2em}
    \caption{Exp. 3. Comparison of two stronger attack models.}
    \label{fig:side_by_side_fixed}
    \vspace{-1em}
\end{figure}

\smallskip
\noindent\underline{\bf Exp. 4: Attacking multiple nodes in a network}
\smallskip

\noindent In this experiment, we assume the attacker captures a fraction ($\lambda$) of the total $d$ nodes, where each node contains $n/d$ local $p$-values, resulting in the same total of $m$ $p$-values under attack. We studied how FDR and power (defined in Section~\ref{sec:form}) behave when we increase $\lambda$ under three attack models: BH-classifier, enhanced BH‑classifier, and shuffling. Note that in this setting, the test statistics for alternative hypothesis are sampled from $N(\mu, 1)$, where $\mu \sim \text{Unif}(2.5, 3.0)$ to better illustrate the change in power. 

\begin{figure}[htbp]
    \centering
    \begin{subfigure}[t]{1.0\linewidth} 
        \centering
        \includegraphics[width=\linewidth]{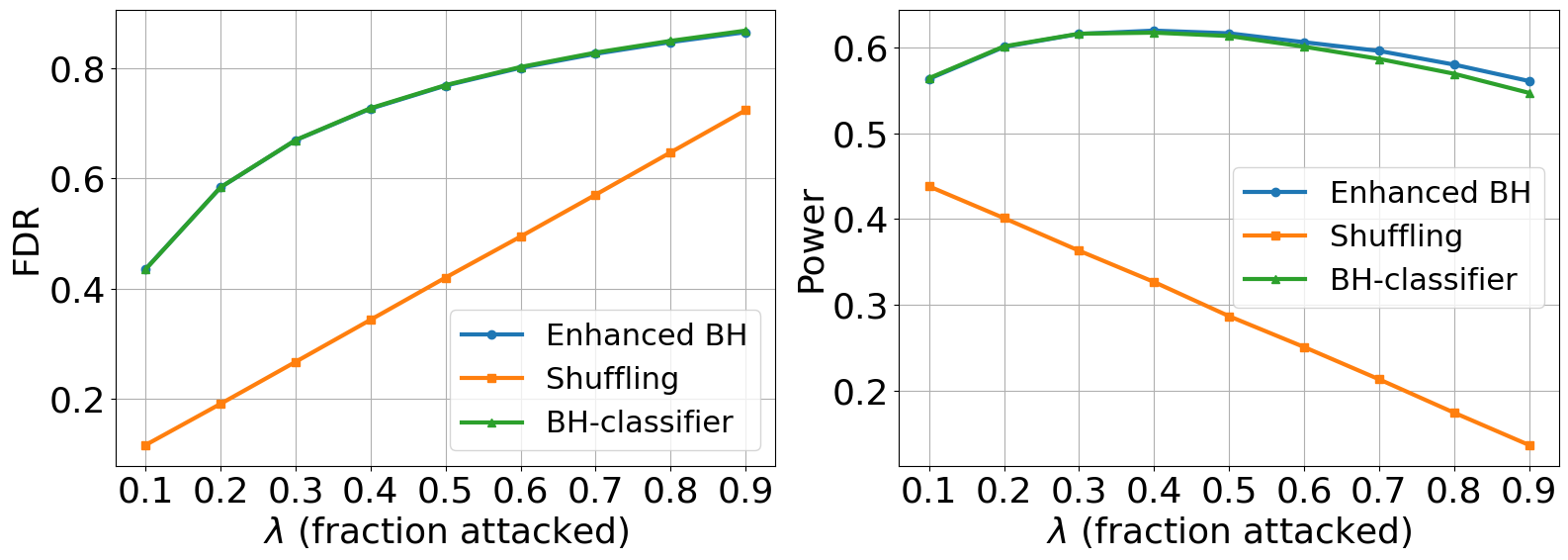} 
        \label{fig:fig1}
    \end{subfigure}
    \hspace{0.05\linewidth} 
\vspace{-1.2em}
    \caption{Exp. 4. Comparison of the three attack models in the distributed setting ($d=20$).}
    \label{fig:side_by_side_fixed}
    \vspace{-1em}
\end{figure}

The results indicates that for the shuffling attack, FDR increases linearly as $\lambda$ increases. While for the enhanced BH-classifier attack, FDR rises much more steeply at small $\lambda$ and then begins to level off as $\lambda$ grows. In other words, the enhanced BH‑classifier attack injects so many low p-values even when only a few nodes are compromised that the global BH procedure already suffers a high FDR; adding more attacked nodes yields only diminishing marginal increases. 

\section{Discussion}
 Our initial studies reported in this work open up several natural and important future directions. When the exact $p$-values are not available at each agent, we will study the impact of Byzantine attacks on empirical $p$-values or more general data-driven score functions (e.g., neural network-based methods~\cite{marandon2024adaptive}). To handle large-scale settings where each agent has a large number of local test statistics, it becomes important to incorporate the resource-efficiency consideration (e.g., with a limited communication budget~\cite{pournaderi2023sample,pournaderi2023large}) into the attack and counter-attack models. Furthermore, the analysis of the detection power is important in providing a comprehensive understanding of different attack models as well as counter-attack strategies. Finally, it would be worthwhile to broaden the class of attack models, drawing inspiration from existing ones (e.g., altering the order of statistics~\cite{quan2023ordered}). 


\section*{Acknowledgment}
This work was supported in part by the National Science Foundation
under Grant CCF-2420146.

    




\newpage
\balance
\bibliography{ref.bib}
\bibliographystyle{IEEEtran}

\end{document}